\documentclass[conference]{IEEEtran}

\usepackage{epsfig,amsmath,amssymb,epsf,amsthm,scalefnt,multirow}
\usepackage{xcolor}
\usepackage{float}
\usepackage{cite}
\usepackage{psfrag}
\usepackage{algorithm}
\usepackage{algpseudocode}\algdef{SE}[While]{NoDoWhile}{EndWhile}[1]{\algorithmicwhile\ #1}{\algorithmicend\ \algorithmicwhile}

\newtheorem{lemma}{Lemma}
\newtheorem*{lemma*}{Lemma}

\begin{document}

\title{Dominant Channel Estimation via MIPS for Large-Scale Antenna Systems with One-Bit ADCs}


\author{\IEEEauthorblockN{In-soo Kim, Namyoon Lee, and Junil Choi}
\IEEEauthorblockA{Department of Electrical Engineering\\
POSTECH\\
Email: \{insookim, nylee, junil\}@postech.ac.kr}}

\maketitle

\begin{abstract}
In large-scale antenna systems, using one-bit analog-to-digital converters (ADCs) has recently become important since they offer significant reductions in both power and cost. However, in contrast to high-resolution ADCs, the coarse quantization of one-bit ADCs results in an irreversible loss of information. In the context of channel estimation, studies have been developed extensively to combat the performance loss incurred by one-bit ADCs. Furthermore, in the field of array signal processing, direction-of-arrival (DOA) estimation combined with one-bit ADCs has gained growing interests recently to minimize the estimation error. In this paper, a channel estimator is proposed for one-bit ADCs where the channels are characterized by their angular geometries, e.g., uniform linear arrays (ULAs). The goal is to estimate the dominant channel among multiple paths. The proposed channel estimator first finds the DOA estimate using the maximum inner product search (MIPS). Then, the channel fading coefficient is estimated using the concavity of the log-likelihood function. The limit inherent in one-bit ADCs is also investigated, which results from the loss of magnitude information.
\end{abstract}

\section{Introduction}
In large-scale antenna systems, known as massive multiple-input multiple-output (MIMO), a significant performance gain is achieved by deploying a large number of antennas at the base station. However, the hardware cost and excessive power consumption due to the large number of antennas make high-resolution analog-to-digital converters (ADCs) less attractive, and using low-resolution, or even the extreme case of one-bit, ADCs has gained popularity \cite{larsson2014massive, bjornson2016massive, risi2014massive, wang2017quantization}. The drawback of one-bit ADCs is that the information loss resulting from the coarse quantization is severe, acting as a performance bottleneck.

In \cite{choi2016near, mo2014channel, li2017channel}, channel estimators for one-bit ADCs have been proposed. In \cite{choi2016near}, the near maximum likelihood (nML) channel estimator relies on the concavity of the log-likelihood function assuming that the channels have no predefined structure. The modified expectation-maximization (EM) algorithm proposed in \cite{mo2014channel} exploits the channel sparsity of millimeter wave systems. However, the computational complexities of the channel estimators introduced in \cite{choi2016near, mo2014channel} are expensive because both are based on the ML channel estimator. The Bussgang linear minimum mean squared error (BLMMSE) channel estimator introduced in \cite{li2017channel} uses the second-order channel statistics to formulate the one-bit LMMSE channel estimator based on the Bussgang decomposition, which is only applicable to Gaussian distributed channels.

In the field of array signal processing, direction-of-arrival (DOA) estimation for one-bit ADCs has been studied in the past \cite{bar2002doa}, although not widely, and begun to be in the spotlight recently \cite{liu2017one, gao2017gridless, stein2016doa}. In \cite{bar2002doa}, extensions of the conventional DOA estimators to one-bit ADCs, based on the arcsine law \cite{van1966spectrum, jacovitti1994estimation}, have been studied. The one-bit spatial smoothing MUSIC (SS MUSIC) proposed in \cite{liu2017one}, which also relies on the arcsine law, considers DOA estimation in sparse arrays. The drawback of the DOA estimators in \cite{bar2002doa, liu2017one} is that many independent observations are needed to reconstruct the covariance matrix. The work in \cite{gao2017gridless} proposed the gridless one-bit DOA estimator based on the support vector machine (SVM), which is cumbersome in practice since the coarse DOA estimate based on the SVM should be refined using the Taylor expansion to obtain the gridless DOA estimate.

In this paper, we propose a channel estimator for one-bit ADCs where the channels are characterized by their angular geometries, e.g., uniform linear arrays (ULAs). In particular, the channel is composed of multiple paths where each path is characterized by its channel fading coefficient and steering vector parameterized by the DOA. The goal is to estimate the dominant channel, known as the line-of-sight (LOS) channel. The proposed channel estimator finds the DOA estimate via the maximum inner product search (MIPS). Then, the channel fading coefficient is estimated by maximizing the log-likelihood function at the MIPS DOA estimate using convex optimization. In contrast to \cite{bar2002doa, liu2017one}, the MIPS DOA estimator finds the DOA estimate using the instantaneous received signal. The simulation results show that the proposed channel estimator performs close to the pseudo ML (pML) channel estimator, which maximizes the likelihood function assuming that the sum of the non-line-of-sight (NLOS) channels and noise is white and Gaussian. The MIPS-based channel estimator is also computationally efficient compared to the pML channel estimator. In addition, we investigate the performance limit of the pML channel estimator in one-bit ADCs.

\section{System Model}
We consider a single-input multiple-output (SIMO) system where the receiver employs an array of antennas characterized by its angular geometry. At each antenna, the received signal's real and imaginary parts are quantized by one-bit ADCs.

The single-antenna transmitter transmits a pilot sequence of length $N$ to the $M$-antenna receiver. Therefore, the received signal $\mathbf{Y}\in\mathbb{C}^{M\times N}$ is
\begin{equation}
\mathbf{Y}=\sqrt{\rho}\mathbf{h}\mathbf{x}^{T}+\mathbf{N}
\end{equation}
where the $m$-th row and $n$-th column correspond to the $m$-th antenna and $n$-th time slot, respectively. The channel $\mathbf{h}\in\mathbb{C}^{M}$ is formed by the combination of the LOS and NLOS channels $\mathbf{h}_{0}$ and $\mathbf{h}_{\ell}$ where $\ell\neq 0$, respectively, i.e.,
\begin{align}
       \mathbf{h}&=\sum_{\ell=0}^{L}c_{\ell}\mathbf{h}_{\ell},\\
\mathbf{h}_{\ell}&=g_{\ell}\mathbf{a}(\theta_{\ell}),
\end{align}
and
\begin{equation}
c_{\ell}=\begin{cases}\sqrt{K/(K+1)}&\text{if }\ell=0\\
                     \sqrt{1/L(K+1)}&\text{if }\ell\neq 0\end{cases},
\end{equation}
where $L$ is the number of NLOS paths, $K$ is the Rician $K$-factor, $g_{\ell}\in\mathbb{C}$ is the zero-mean and unit-variance $\ell$-th channel fading coefficient, $\theta_{\ell}$ is the $\ell$-th DOA, and $\mathbf{a}(\theta_{\ell})\in\mathbb{C}^{M}$ is the steering vector parameterized by $\theta_{\ell}$. In particular, we assume that $\mathbf{a}(\theta_{\ell})$ obeys the power constraint of $|a_{m}(\theta_{\ell})|^{2}=1$ where $a_{m}(\theta_{\ell})$ is the $m$-th element of $\mathbf{a}(\theta_{\ell})$. For example, if the receiver is modeled as a ULA where the inter-element spacing is $d$ while the wavelength is $\lambda$, the steering vector is
\begin{equation}\label{ULA}
\mathbf{a}(\theta_{\ell})=\begin{bmatrix}1, e^{-j\frac{2\pi d}{\lambda}\sin{\theta_{\ell}}}, \cdots, e^{-j\frac{2\pi d}{\lambda}(M-1)\sin{\theta_{\ell}}}\end{bmatrix}^{T}.
\end{equation}
The support of $g_{\ell}$ and $\theta_{\ell}$ are $\mathbb{C}$ and $\Theta\subset\mathbb{R}$, respectively. The pilot sequence $\mathbf{x}=\begin{bmatrix}x_{1}, \cdots, x_{N}\end{bmatrix}^{T}\in\mathbb{C}^{N}$ of length $N$ follows $|x_{n}|^{2}=1$ to simplify the peak power constraint. The elements of the additive white Gaussian noise (AWGN) $\mathbf{N}\in\mathbb{C}^{M\times N}$ are independent and identically distributed (i.i.d.) as $\mathcal{CN}(0,1)$. Furthermore, $g_{\ell}$, $\theta_{\ell}$, and $\mathbf{N}$ are assumed to be independent. The signal-to-noise ratio (SNR) is defined as $\rho$.

The vectorized received signal $\mathbf{y}=\mathrm{vec}(\mathbf{Y})$ is
\begin{equation}
\mathbf{y}=\sqrt{\rho}\mathbf{X}\mathbf{h}+\mathbf{n}
\end{equation}
where $\mathbf{X}=\mathbf{x}\otimes\mathbf{I}_{M}$ and $\mathbf{n}=\mathrm{vec}(\mathbf{N})\sim\mathcal{CN}(\mathbf{0}_{MN},\mathbf{I}_{MN})$. At each antenna, the real and imaginary parts of $\mathbf{y}$ are quantized by one-bit ADCs. The quantized received signal $\hat{\mathbf{y}}$ is
\begin{equation}
\hat{\mathbf{y}}=\mathrm{Q}(\mathbf{y})
\end{equation}
where $\mathrm{Q}(\cdot)$ is the element-wise one-bit quantization function. In this paper, zero thresholds are used, i.e.,
\begin{equation}
\mathrm{Q}(\mathbf{y})=\frac{1}{\sqrt{2}}\mathrm{sgn}(\mathbf{y})
\end{equation}
where $\mathrm{sgn}(\cdot)$ is the sign function, which is applied to the real and imaginary parts element-wise. Therefore, the elements of $\hat{\mathbf{y}}$ are constrained to the quadrature phase-shift keying (QPSK) constellation points, i.e.,
\begin{equation}
\hat{y}_{k}\in\frac{1}{\sqrt{2}}\{1+j, 1-j, -1+j, -1-j\}
\end{equation}
where $\hat{y}_{k}$ is the $k$-th element of $\hat{\mathbf{y}}$. The goal is to estimate $\mathbf{h}_{0}$ from $\hat{\mathbf{y}}$.

\textbf{Remark 1:} In this paper, the steering vector is assumed to be parameterized by the one-dimensional DOA to simplify the notations. However, the MIPS-based channel estimator can be applied to any angular geometry where the steering vector is parameterized by the multi-dimensional DOA as long as the elements of the steering vector obey the unit-power constraint, e.g., uniform planar arrays (UPAs) that are parameterized by the horizontal and vertical DOAs.

\section{pML Channel Estimation}
In this section, we analyze the pML channel estimator. The observations established in this section builds the framework of the MIPS-based channel estimator in Section \ref{section_4}. To develop the pML channel estimator, we write $\mathbf{y}$ as
\begin{align}
\mathbf{y}&=\sqrt{\rho}\mathbf{X}c_{0}\mathbf{h}_{0}+\sqrt{\rho}\mathbf{X}\sum_{\ell=1}^{L}c_{\ell}\mathbf{h}_{\ell}+\mathbf{n}\notag\\
          &=\sqrt{\rho}\mathbf{X}c_{0}\mathbf{h}_{0}+\bar{\mathbf{n}}.
\end{align}
In general, it is hard to estimate $\mathbf{h}_{0}$ because $\bar{\mathbf{n}}$ is neither white nor Gaussian. Hence, the pML channel estimator circumvents this problem by approximating $\bar{\mathbf{n}}$ by $\tilde{\mathbf{n}}$, which is distributed as $\mathcal{CN}(\mathbb{E}\{\bar{\mathbf{n}}\},\mathrm{diag}(\mathbb{E}\{\bar{\mathbf{n}}\bar{\mathbf{n}}^{H}\}))$, i.e,
\begin{equation}
\tilde{\mathbf{n}}\sim\mathcal{CN}(\mathbf{0}_{MN},\sigma^{2}\mathbf{I}_{MN})
\end{equation}
where
\begin{equation}
\sigma^{2}=\begin{cases}\rho/(K+1)+1&\text{if }L\neq 0\\
                                   1&\text{if }L=0\end{cases},
\end{equation}
and maximizes the likelihood function to estimate $\mathbf{h}_{0}$. In fact, the pML channel estimator reduces to the ML channel estimator when $L=0$ because $\bar{\mathbf{n}}$ is both white and Gaussian. To formulate the likelihood function assuming $\tilde{\mathbf{n}}$, as a preliminary step, we define $\tilde{\rho}=\rho/\sigma^{2}$. The $k$-th element of $\mathbf{X}c_{0}\mathbf{a}(\theta_{0})$ is denoted by $X_{k}(\theta_{0})$. In addition, define $\mathbf{f}_{R,k}(\theta_{0})$, $\mathbf{f}_{I,k}(\theta_{0})$, and $\mathbf{g}_{0}$ as
\begin{align}
\mathbf{f}_{R,k}(\theta_{0})&=\begin{bmatrix}\mathrm{Re}(X_{k}(\theta_{0})), -\mathrm{Im}(X_{k}(\theta_{0}))\end{bmatrix}^{T},\\
\mathbf{f}_{I,k}(\theta_{0})&=\begin{bmatrix}\mathrm{Im}(X_{k}(\theta_{0})), \mathrm{Re}(X_{k}(\theta_{0}))\end{bmatrix}^{T},\\
              \mathbf{g}_{0}&=\begin{bmatrix}\mathrm{Re}(g_{0}), \mathrm{Im}(g_{0})\end{bmatrix}^{T}.
\end{align}
Then, assuming $\tilde{\mathbf{n}}$, the $k$-th element of $\mathbf{y}$ is distributed as
\begin{equation}
\mathcal{CN}(\sqrt{\rho}X_{k}(\theta_{0})g_{0},\sigma^{2})
\end{equation}
conditioned on $g_{0}$ and $\theta_{0}$. Thus, because the elements of $\tilde{\mathbf{n}}$ are independent whose real and imaginary parts are independent, we obtain the log-likelihood function $L_{\tilde{\rho}}(\mathbf{g}_{0}', \theta_{0}')$ evaluated at $\mathbf{g}_{0}'\in\mathbb{R}^{2}$ and $\theta_{0}'\in\Theta$, which is
\begin{align}\label{log_likelihood_function}
L_{\tilde{\rho}}(\mathbf{g}_{0}', \theta_{0}')=&\sum_{k=1}^{MN}(\log\Phi(2\hat{y}_{R,k}\sqrt{\tilde{\rho}}\mathbf{f}_{R,k}^{T}(\theta_{0}')\mathbf{g}_{0}')\notag\\
                                               &+\log\Phi(2\hat{y}_{I,k}\sqrt{\tilde{\rho}}\mathbf{f}_{I,k}^{T}(\theta_{0}')\mathbf{g}_{0}'))
\end{align}
where $\hat{y}_{R,k}=\mathrm{Re}(\hat{y}_{k})$, $\hat{y}_{I,k}=\mathrm{Im}(\hat{y}_{k})$, and $\Phi(\cdot)$ represents the cumulative distribution function (CDF) of $\mathcal{N}(0,1)$. Then, the pML channel estimate $\hat{\mathbf{h}}_{0, \mathrm{pML}}$ of $\mathbf{h}_{0}$ is defined as
\begin{equation}
\hat{\mathbf{h}}_{0, \mathrm{pML}}=\hat{g}_{0, \mathrm{pML}}\mathbf{a}(\hat{\theta}_{0, \mathrm{pML}})
\end{equation}
where
\begin{equation}\label{optimization_problem}
(\hat{\mathbf{g}}_{0, \mathrm{pML}}, \hat{\theta}_{0, \mathrm{pML}})=\underset{\mathbf{g}_{0}'\in\mathbb{R}^{2}, \theta_{0}'\in\Theta}{\mathrm{argmax}}L_{\tilde{\rho}}(\mathbf{g}_{0}', \theta_{0}').
\end{equation}
$\hat{\mathbf{g}}_{0, \mathrm{pML}}$ contains the real and imaginary parts of $\hat{g}_{0, \mathrm{pML}}$.

In general, solving \eqref{optimization_problem} is computationally cumbersome due to the vast size of the search space and expensive objective function. To find $\hat{\mathbf{h}}_{0, \mathrm{pML}}$ more efficiently, we investigate the structure of $L_{\tilde{\rho}}(\mathbf{g}_{0}', \theta_{0}')$. First, note that $L_{\tilde{\rho}}(\mathbf{g}_{0}', \theta_{0}')$ is a concave function of $\mathbf{g}_{0}'$ because the argument of each $\Phi(\cdot)$ is an affine function of $\mathbf{g}_{0}'$, which preserves the log-concavity of $\Phi(\cdot)$ \cite{boyd2004convex}; a similar observation was established in \cite{choi2016near}. Thus, $\hat{\mathbf{h}}_{0, \mathrm{pML}}$ can be found more efficiently by interpreting \eqref{optimization_problem} as
\begin{equation}\label{two_optimization_problem}
\underset{\mathbf{g}_{0}'\in\mathbb{R}^{2}, \theta_{0}'\in\Theta}{\mathrm{max}}L_{\tilde{\rho}}(\mathbf{g}_{0}', \theta_{0}')=\underset{\theta_{0}'\in\Theta}{\mathrm{max}}\underset{\mathbf{g}_{0}'\in\mathbb{R}^{2}}{\mathrm{max}}L_{\tilde{\rho}}(\mathbf{g}_{0}', \theta_{0}')
\end{equation}
where the inner optimization problem of \eqref{two_optimization_problem} can be solved using convex optimization, e.g., the gradient descent method (GDM) using the backtracking line search \cite{boyd2004convex}.

Second, note that $L_{\tilde{\rho}}(\mathbf{g}_{0}', \theta_{0}')$ has no explicit 
structure with respect to $\theta_{0}'$. Thus, to solve the outer optimization problem of \eqref{two_optimization_problem}, $\Theta$ should be searched exhaustively; that is, the number of convex optimization performed is proportional to the size of $\Theta$. For example, assume that the inner optimization problem is solved using the GDM while the outer optimization problem is solved using the $B$-bit grid search, thereby partitioning $\Theta$ into $2^{B}$ uniform grid points. Then, the pML channel estimator performs the GDM $2^{B}$ times. In this paper, we assume that the pML channel estimator solves the inner and outer optimization problems of \eqref{two_optimization_problem} using the GDM, which uses the backtracking line search, and $B$-bit grid search, respectively. In Section \ref{section_5}, the computational complexity of the pML channel estimator is investigated by counting the number of real multiplications performed. The explanation of how the number of real multiplications performed by the pML channel estimator is counted is omitted because of the page limit.

Before we proceed to the MIPS-based channel estimator, we investigate the fundamental limit of one-bit ADCs in the high SNR regime through the lens of the pML channel estimator.
\begin{lemma}\label{lemma_1}
$\|\hat{\mathbf{h}}_{0, \mathrm{pML}}\|\to 0$ as $\tilde{\rho}\to\infty$ almost surely.
\end{lemma}
\begin{proof}
The goal of the proof is to show that $\|\hat{\mathbf{g}}_{0,\mathrm{pML}}\|\to 0$ as $\tilde{\rho}\to\infty$ almost surely because $\|\hat{\mathbf{h}}_{0, \mathrm{pML}}\|=\sqrt{M}\|\hat{\mathbf{g}}_{0, \mathrm{pML}}\|$. Therefore, we proceed by showing that $\|\hat{\mathbf{g}}_{\tilde{\rho}}(\theta_{0}')\|\to 0$ as $\tilde{\rho}\to\infty$ for any realization of $\hat{\mathbf{y}}$ evaluated at any $\theta_{0}'$ where $\hat{\mathbf{g}}_{\tilde{\rho}}(\theta_{0}')$ is defined as
\begin{equation}
\hat{\mathbf{g}}_{\tilde{\rho}}(\theta_{0}')=\underset{\mathbf{g}_{0}'\in\mathbb{R}^{2}}{\mathrm{argmax}}L_{\tilde{\rho}}(\mathbf{g}_{0}', \theta_{0}').
\end{equation}
First, consider $\tilde{\rho}_{1}>0$ and $\tilde{\rho}_{2}=k\tilde{\rho}_{1}$ where $k>0$. From \eqref{log_likelihood_function}, observe that $L_{\tilde{\rho}_{2}}(\mathbf{g}_{0}', \theta_{0}')=L_{\tilde{\rho}_{1}}(\tilde{\mathbf{g}}_{0}, \theta_{0}')$ where $\tilde{\mathbf{g}}_{0}=\sqrt{k}\mathbf{g}_{0}'$. Then, since
\begin{align}
\hat{\mathbf{g}}_{\tilde{\rho}_{2}}(\theta_{0}')&\overset{\hphantom{(a)}}{=}\underset{\mathbf{g}_{0}'\in\mathbb{R}^{2}}{\mathrm{argmax}}L_{\tilde{\rho}_{2}}(\mathbf{g}_{0}', \theta_{0}')\notag\\
                                                &\overset{(a)}{=}\frac{1}{\sqrt{k}}\underset{\tilde{\mathbf{g}}_{0}\in\mathbb{R}^{2}}{\mathrm{argmax}}L_{\tilde{\rho}_{1}}(\tilde{\mathbf{g}}_{0}, \theta_{0}')\notag\\
                                                &\overset{\hphantom{(a)}}{=}\frac{1}{\sqrt{k}}\hat{\mathbf{g}}_{\tilde{\rho}_{1}}(\theta_{0}')
\end{align}
where (a) results from substituting $\tilde{\mathbf{g}}_{0}$ for $\sqrt{k}\mathbf{g}_{0}'$, we observe that $\|\hat{\mathbf{g}}_{\tilde{\rho}_{2}}(\theta_{0}')\|\to 0$ as $k\to\infty$, which is equivalent to saying that $\|\hat{\mathbf{g}}_{\tilde{\rho}}(\theta_{0}')\|\to 0$ as $\tilde{\rho}\to\infty$, for any realization of $\hat{\mathbf{y}}$ evaluated at any $\theta_{0}'$. Hence, from $\hat{\mathbf{g}}_{0, \mathrm{pML}}=\hat{\mathbf{g}}_{\tilde{\rho}}(\hat{\theta}_{0, \mathrm{pML}})$, we note that $\|\hat{\mathbf{g}}_{0, \mathrm{pML}}\|\to 0$ as $\tilde{\rho}\to\infty$ for any realization of $\hat{\mathbf{y}}$, arriving at
\begin{align}
\mathrm{Pr}\left[\lim_{\tilde{\rho}\to\infty}\|\hat{\mathbf{h}}_{0, \mathrm{pML}}\|=0\right]&=\mathrm{Pr}\left[\lim_{\tilde{\rho}\to\infty}\|\hat{\mathbf{g}}_{0, \mathrm{pML}}\|=0\right]\notag\\
                                                                                            &=1,
\end{align}
which completes the proof.
\end{proof}
Lemma \ref{lemma_1} shows that $\|\hat{\mathbf{h}}_{0, \mathrm{pML}}\|$ becomes deterministic in the high SNR regime, implying that one-bit ADCs do not convey any channel quality information (CQI) in the absence of noise. In contrast to full-resolution ADCs where noise is generally not welcomed, the presence of noise may enhance the channel estimation performance in one-bit ADCs. The intuitive reason why the magnitude information of the received signal is lost in one-bit ADCs can be explained by noting that
\begin{align}\label{magnitude_information}
\mathrm{Q}(\sqrt{\rho}\mathbf{X}c_{0}\mathbf{h}_{0}+\tilde{\mathbf{n}})&\approx\mathrm{Q}(\sqrt{\rho}\mathbf{X}c_{0}\mathbf{h}_{0})\notag\\
                                                                       &=\mathrm{Q}(k\sqrt{\rho}\mathbf{X}c_{0}\mathbf{h}_{0})
\end{align}
in the high SNR regime where $k>0$. According to \eqref{magnitude_information}, the information embedded in $k$ is indistinguishable in the high SNR regime because one-bit ADCs do not convey any magnitude information. Hence, we expect that the performance of the pML channel estimator falls as the SNR enters the high SNR regime.

\section{MIPS-Based Channel Estimation}\label{section_4}
In this section, we derive the MIPS-based channel estimator. The problem of the pML channel estimator is that $L_{\tilde{\rho}}(\mathbf{g}_{0}', \theta_{0}')$ has no explicit structure with respect to $\theta_{0}'$. Thus, the number of convex optimization performed is proportional to the size of $\Theta$. The MIPS-based channel estimator attempts to break this computational bottleneck by performing two-stage channel estimation, which separates the search space into $\Theta$ and $\mathbb{R}^{2}$. In this section, similar to the pML channel estimator, the MIPS-based channel estimator considers $\tilde{\mathbf{n}}$ instead of $\bar{\mathbf{n}}$.

\subsection{DOA Estimator}
In this subsection, the first stage of the MIPS-based channel estimator is described. In the first stage, $\theta_{0}$ is estimated in the search space $\Theta$. Thus, we can interpret the problem as DOA estimation in one-bit ADCs. In \cite{bar2002doa}, DOA estimators for one-bit ADCs were proposed using the arcsine law, which relates the covariance matrices of the quantized and unquantized received signals using arcsine \cite{van1966spectrum, jacovitti1994estimation}. In this paper, motivated by the approach of \cite{bar2002doa}, we propose the MIPS DOA estimate of $\theta_{0}$.

To derive the MIPS DOA estimate of $\theta_{0}$, we express $\theta_{0}$ in terms of the conditional covariance matrix of $\hat{\mathbf{y}}$ given $\theta_{0}$. The conditional covariance matrix of $\mathbf{y}$ given $\theta_{0}$, namely $\mathbf{C}_{\mathbf{y}}(\theta_{0})$, is
\begin{align}\label{covariance_matrix}
\mathbf{C}_{\mathbf{y}}(\theta_{0})&=\mathbb{E}\{(\sqrt{\rho}\mathbf{X}c_{0}\mathbf{h}_{0}+\tilde{\mathbf{n}})(\sqrt{\rho}\mathbf{X}c_{0}\mathbf{h}_{0}+\tilde{\mathbf{n}})^{H}|\theta_{0}\}\notag\\
                                   &=\rho c_{0}^{2}\mathbf{X}\mathbf{a}(\theta_{0})\mathbf{a}(\theta_{0})^{H}\mathbf{X}^{H}+\sigma^{2}\mathbf{I}_{MN}
\end{align}
where $\bar{\mathbf{n}}$ was approximated by $\tilde{\mathbf{n}}$. The diagonal matrix formed by $\mathbf{C}_{\mathbf{y}}(\theta_{0})$ is denoted by $\mathbf{\Sigma}_{\mathbf{y}}(\theta_{0})$, which is
\begin{align}
\mathbf{\Sigma}_{\mathbf{y}}(\theta_{0})&=\mathrm{diag}(\mathbf{C}_{\mathbf{y}}(\theta_{0}))\notag\\
                                        &=(\rho c_{0}^{2}+\sigma^{2})\mathbf{I}_{MN}.
\end{align}
Then, from the arcsine law, $\mathbf{C}_{\hat{\mathbf{y}}}(\theta_{0})$, which is the conditional covariance matrix of $\hat{\mathbf{y}}$ given $\theta_{0}$, can be expressed by $\mathbf{C}_{\mathbf{y}}(\theta_{0})$ as
\begin{equation}
\mathbf{C}_{\hat{\mathbf{y}}}(\theta_{0})=\frac{2}{\pi}\arcsin\left(\mathbf{\Sigma}_{\mathbf{y}}^{-\frac{1}{2}}(\theta_{0})\mathbf{C}_{\mathbf{y}}(\theta_{0})\mathbf{\Sigma}_{\mathbf{y}}^{-\frac{1}{2}}(\theta_{0})\right)
\end{equation}
and vice versa, i.e.,
\begin{align}\label{sine_law}
\mathbf{C}_{\mathbf{y}}(\theta_{0})&=\mathbf{\Sigma}_{\mathbf{y}}^{\frac{1}{2}}(\theta_{0})\sin\left(\frac{\pi}{2}\mathbf{C}_{\hat{\mathbf{y}}}(\theta_{0})\right)\mathbf{\Sigma}_{\mathbf{y}}^{\frac{1}{2}}(\theta_{0})\notag\\
                                   &=(\rho c_{0}^{2}+\sigma^{2})\sin\left(\frac{\pi}{2}\mathbf{C}_{\hat{\mathbf{y}}}(\theta_{0})\right)
\end{align}
where $\arcsin(\cdot)$ and $\sin(\cdot)$ are the element-wise arcsine and sine functions applied to the real and imaginary parts, respectively. In addition, based on the conventional beamformer approach, which attempts to maximize the output power \cite{krim1996two}, $\theta_{0}$ can be written as
\begin{align}\label{conventional_beamformer}
                           &\underset{\theta_{0}'\in\Theta}{\mathrm{argmax}}(\mathbf{X}\mathbf{a}(\theta_{0}'))^{H}\mathbf{C}_{\mathbf{y}}(\theta_{0})\mathbf{X}\mathbf{a}(\theta_{0}')\notag\\
           \overset{(a)}{=}&\underset{\theta_{0}'\in\Theta}{\mathrm{argmax}}(\rho c_{0}^{2}|(\mathbf{X}\mathbf{a}(\theta_{0}'))^{H}\mathbf{X}\mathbf{a}(\theta_{0})|^{2}+\sigma^{2}\|\mathbf{X}\mathbf{a}(\theta_{0}')\|^{2})\notag\\
           \overset{(b)}{=}&\underset{\theta_{0}'\in\Theta}{\mathrm{argmax}}(\rho c_{0}^{2}N^{2}|\mathbf{a}(\theta_{0}')^{H}\mathbf{a}(\theta_{0})|^{2}+\sigma^{2} MN)\notag\\
\overset{\hphantom{(a)}}{=}&\underset{\theta_{0}'\in\Theta}{\mathrm{argmax}}|\mathbf{a}(\theta_{0}')^{H}\mathbf{a}(\theta_{0})|^{2}\notag\\
           \overset{(c)}{=}&\theta_{0}
\end{align}
where (a), (b), and (c) are the consequences of \eqref{covariance_matrix}, $\|\mathbf{x}\|^{2}=N$, and the unit-power constraint imposed on the steering vector's elements, respectively. At this point, $\theta_{0}$ can be expressed in terms of $\mathbf{C}_{\hat{\mathbf{y}}}(\theta_{0})$ by putting \eqref{sine_law} into \eqref{conventional_beamformer}, i.e.,
\begin{equation}\label{theta}
\theta_{0}=\underset{\theta_{0}'\in\Theta}{\mathrm{argmax}}(\mathbf{X}\mathbf{a}(\theta_{0}'))^{H}\sin\left(\frac{\pi}{2}\mathbf{C}_{\hat{\mathbf{y}}}(\theta_{0})\right)\mathbf{X}\mathbf{a}(\theta_{0}').
\end{equation}
\alglanguage{pseudocode}
\begin{algorithm}[tbp]
	\caption{MIPS-based channel estimator}
	\begin{algorithmic}[1]
		\State Find $\hat{\theta}_{0, \mathrm{MIPS}}=\underset{\theta_{0}'\in\Theta}{\mathrm{argmax}}|(\mathbf{X}\mathbf{a}(\theta_{0}'))^{H}\hat{\mathbf{y}}|$
		\State Find $\hat{\mathbf{g}}_{0, \mathrm{MIPS}}=\underset{\mathbf{g}_{0}'\in\mathbb{R}^{2}}{\mathrm{argmax}}L_{\tilde{\rho}}(\mathbf{g}_{0}', \hat{\theta}_{0, \mathrm{MIPS}})$
		\State $\hat{\mathbf{h}}_{0, \mathrm{MIPS}}=\hat{g}_{0, \mathrm{MIPS}}\mathbf{a}(\hat{\theta}_{0, \mathrm{MIPS}})$
	\end{algorithmic}
\end{algorithm}

However, since the receiver has no prior knowledge of $\mathbf{C}_{\hat{\mathbf{y}}}(\theta_{0})$, we use the sample covariance matrix $\hat{\mathbf{C}}_{\hat{\mathbf{y}}}$, which is
\begin{equation}
\hat{\mathbf{C}}_{\hat{\mathbf{y}}}=\hat{\mathbf{y}}\hat{\mathbf{y}}^{H},
\end{equation}
to develop the DOA estimator. By replacing $\mathbf{C}_{\hat{\mathbf{y}}}(\theta_{0})$ with $\hat{\mathbf{C}}_{\hat{\mathbf{y}}}$ in \eqref{theta}, we obtain the MIPS DOA estimate $\hat{\theta}_{0, \mathrm{MIPS}}$ of $\theta_{0}$, which is
\begin{align}\label{MIPS_DOA_estimate}
\hat{\theta}_{0,\mathrm{MIPS}}&\overset{\hphantom{(a)}}{=}\underset{\theta_{0}'\in\Theta}{\mathrm{argmax}}(\mathbf{X}\mathbf{a}(\theta_{0}'))^{H}\sin\left(\frac{\pi}{2}\hat{\mathbf{C}}_{\hat{\mathbf{y}}}\right)\mathbf{X}\mathbf{a}(\theta_{0}')\notag\\
                              &\overset{\hphantom{(a)}}{=}\underset{\theta_{0}'\in\Theta}{\mathrm{argmax}}(\mathbf{X}\mathbf{a}(\theta_{0}'))^{H}\sin\left(\frac{\pi}{2}\hat{\mathbf{y}}\hat{\mathbf{y}}^{H}\right)\mathbf{X}\mathbf{a}(\theta_{0}')\notag\\
                              &\overset{(a)}{=}\underset{\theta_{0}'\in\Theta}{\mathrm{argmax}}(\mathbf{X}\mathbf{a}(\theta_{0}'))^{H}\hat{\mathbf{y}}\hat{\mathbf{y}}^{H}\mathbf{X}\mathbf{a}(\theta_{0}')\notag\\
                              &\overset{\hphantom{(a)}}{=}\underset{\theta_{0}'\in\Theta}{\mathrm{argmax}}|(\mathbf{X}\mathbf{a}(\theta_{0}'))^{H}\hat{\mathbf{y}}|
\end{align}
where (a) is a direct consequence of the fact that the elements of $\hat{\mathbf{y}}\hat{\mathbf{y}}^{H}$ are constrained to $\{1, -1, j, -j\}$. The reason why the proposed DOA estimator is named the MIPS DOA estimator follows from \eqref{MIPS_DOA_estimate}; the receiver estimates $\theta_{0}$ by searching for the steering vector which maximizes the inner product. To solve \eqref{MIPS_DOA_estimate}, an exhaustive search is required.

\subsection{Channel Fading Coefficient Estimator}
In the second stage, $g_{0}$ is estimated. The MIPS-based channel fading coefficient estimate $\hat{g}_{0, \mathrm{MIPS}}$ of $g_{0}$ is defined as
\begin{equation}\label{MIPS_channel_fading_coefficient_estimate}
\hat{\mathbf{g}}_{0, \mathrm{MIPS}}=\underset{\mathbf{g}_{0}'\in\mathbb{R}^{2}}{\mathrm{argmax}}L_{\tilde{\rho}}(\mathbf{g}_{0}', \hat{\theta}_{0, \mathrm{MIPS}})
\end{equation}
where $\hat{\mathbf{g}}_{0, \mathrm{MIPS}}$ contains the real and imaginary parts of $\hat{g}_{0, \mathrm{MIPS}}$. From \eqref{MIPS_channel_fading_coefficient_estimate}, observe that $\hat{\mathbf{g}}_{0, \mathrm{MIPS}}$ is the maximizer of the inner optimization problem of \eqref{two_optimization_problem} defined at $\hat{\theta}_{0, \mathrm{MIPS}}$, which is a convex optimization problem. Therefore, we can find $\hat{\mathbf{g}}_{0, \mathrm{MIPS}}$ efficiently, which completes the estimation of $\mathbf{h}_{0}$. The MIPS-based channel estimate $\hat{\mathbf{h}}_{0, \mathrm{MIPS}}$ of $\mathbf{h}_{0}$ is defined as
\begin{equation}
\hat{\mathbf{h}}_{0, \mathrm{MIPS}}=\hat{g}_{0, \mathrm{MIPS}}\mathbf{a}(\hat{\theta}_{0, \mathrm{MIPS}}).
\end{equation}
In Algorithm 1, we give an outline of the MIPS-based channel estimator. The first stage requires an exhaustive search over $\Theta$. The second stage can be solved using convex optimization for a given $\hat{\theta}_{0, \mathrm{MIPS}}$. In this paper, we assume that the optimization problems in the first and second stages of the MIPS-based channel estimator are solved by the $B$-bit grid search and GDM where the backtracking line search is employed, respectively.
\begin{figure}[htbp]
	\centering
	\includegraphics[width=0.45\textwidth]{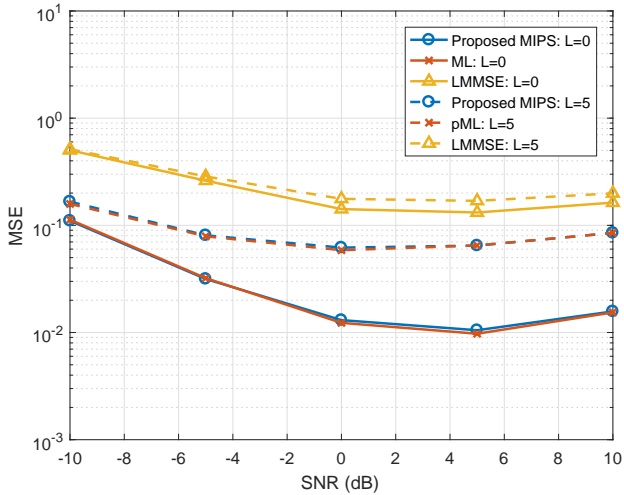}
	\caption{MSE vs. SNR for $M=24$, $N=15$ at different $L$.}
\end{figure}

In Section \ref{section_5}, it is shown that the computational complexity of the MIPS-based channel estimator is in the order of $1/2^{B}$ of the pML channel estimator's since the GDM is performed $2^{B}$ times by the pML channel estimator. In addition, the simulation results verify that the MIPS-based channel estimator performs close to the pML channel estimator. Thus, the MIPS-based channel estimator is computationally efficient while the performance loss is negligible compared to the pML channel estimator.

\textbf{Remark 2:} The performance of the MIPS-based channel estimator is expected to degrade in the high SNR regime because $\hat{\mathbf{g}}_{0, \mathrm{MIPS}}$ shrinks to a meaningless estimate $\mathbf{0}_{2}$ according to Lemma \ref{lemma_1}. In fact, the poor performance of one-bit ADCs in the high SNR regime is a well-known phenomenon, which is inevitable due to the loss of magnitude information.

\section{Simulation Results}\label{section_5}
In this section, Monte-Carlo simulations are performed to evaluate the performance of the MIPS-based channel estimator. It is assumed that the first and second stages of Algorithm 1 are solved by the $B$-bit grid search and GDM, respectively. In addition, we assume that the pML channel estimator finds $\hat{\mathbf{h}}_{0, \mathrm{pML}}$ by solving the inner and outer optimization problems of \eqref{two_optimization_problem} using the GDM and $B$-bit grid search, respectively. The parameter of the $B$-bit grid search is set to $B=8$. The stopping critetion $\eta$ and backtracking line search parameters of the GDM in \cite{boyd2004convex} are set to $\eta=0.01$ and $\alpha=0.1$, $\beta=0.5$, respectively. The zero-forcing (ZF) channel estimate in \cite{choi2016near} is used as the starting point of the GDM.

The Rician $K$-factor is set to 13.5 dB based on the measurements in \cite{muhi2010modelling}, the channel fading coefficients and DOAs are distributed as $g_{\ell}\sim\mathcal{CN}(0,1)$ and $\theta_{\ell}\sim\mathrm{Unif}(\Theta)$, respectively, where $\Theta=[-\pi/3, \pi/3]$, the receiver is modeled as a ULA in \eqref{ULA} where $d=\lambda/2$, and the pilot sequence is selected as the last column of the size $N$ discrete Fourier Transform (DFT) matrix.
\begin{table}[htbp]
	\centering
	\caption{The average number of real multiplications for $M=24$, $N=15$, and $L=5$ at different SNRs.}
	\renewcommand{\arraystretch}{1.5}
	\begin{tabular}{c|c|c}
		\hline
		SNR (dB)&MIPS$\cdot 10^{-6}$&pML$\cdot 10^{-6}$\\
		\hline\hline
		-10&1.8&514.7\\
		\hline
		-5&2.3&313.0\\
		\hline
		0&3.6&1873.5\\
		\hline
		5&5.1&684.3\\
		\hline
		10&13.8&1783.9\\
		\hline
	\end{tabular}
\end{table}

In Fig. 1, we compare the MSEs of the MIPS-based, pML, and LMMSE channel estimators when $M=24$, $N=15$ at different SNRs with $L=0$ and $L=5$. The MSE is defined as
\begin{equation}
\mathrm{MSE}=\frac{1}{M}\mathbb{E}\{\|\hat{\mathbf{h}}_{0}-\mathbf{h}_{0}\|^{2}\}
\end{equation}
where 1000 Monte-Carlo simulations are performed to compute the MSEs of the MIPS-based and pML channel estimators. The covariance matrices needed to compute the LMMSE are found by $10^{6}$ Monte-Carlo simulations. When $L=0$, the pML channel estimator is equivalent to the ML channel estimator because it is the maximizer of the likelihood function; that is, the pML channel estimator is optimal in the sense that the likelihood function is maximized. The interesting point is that the difference between the MSEs of the MIPS-based and ML channel estimators is negligible, which shows that the MIPS-based channel estimator achieves the performance of the ML channel estimator when $L=0$.

When $L=5$, the pML channel estimator becomes suboptimal because it is the maximizer of the approximated likelihood function, which is obtained by assuming that $\bar{\mathbf{n}}$ is both white and Gaussian. The MSE of the MIPS-based channel estimator is close to the pML channel estimator's, which shows that the MIPS-based channel estimator performs as good as the pML channel estimator regardless of $L$. In addition, both the MIPS-based and pML channel estimators perform worse in the high SNR regime, which can be explained by Lemma 1. The LMMSE channel estimator is shown as a comparison.

Table I shows the computational complexities of the MIPS-based and pML channel estimators when $M=24$, $N=15$, and $L=5$ at different SNRs. The measure of the computational complexity is the average number of real multiplications performed. The specific explanation of how they were counted is omitted because of the page limit. The ratio of the average number of real multiplications performed by the MIPS-based channel estimator to the pML channel estimator's is at most 0.7\%, which is in the order of $1/2^{B}$. Therefore, the MIPS-based channel estimator is computationally efficient while the performance gap is negligible compared to the pML channel estimator.
\begin{figure}[htbp]
	\centering
	\includegraphics[width=0.45\textwidth]{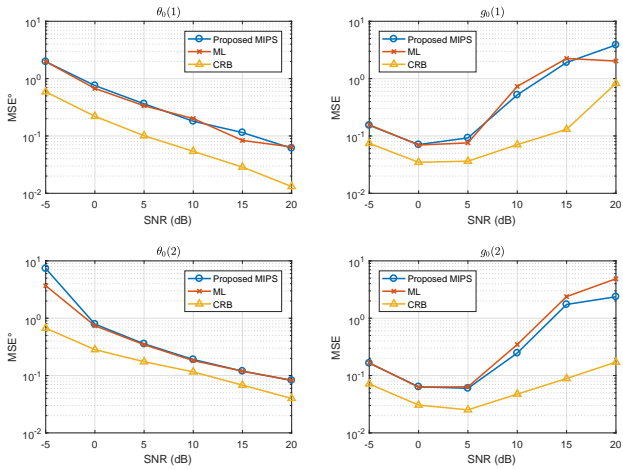}
	\caption{$\mathrm{MSE}(\theta_{0}(k))$ and $\mathrm{MSE}(g_{0}(k))$ vs. SNR for $M=8$, $N=10$, and $L=0$. Two realizations of $\theta_{0}$ and $g_{0}$ are given with the corresponding CRBs.}
\end{figure}
\begin{figure}[htbp]
	\centering
	\includegraphics[width=0.45\textwidth]{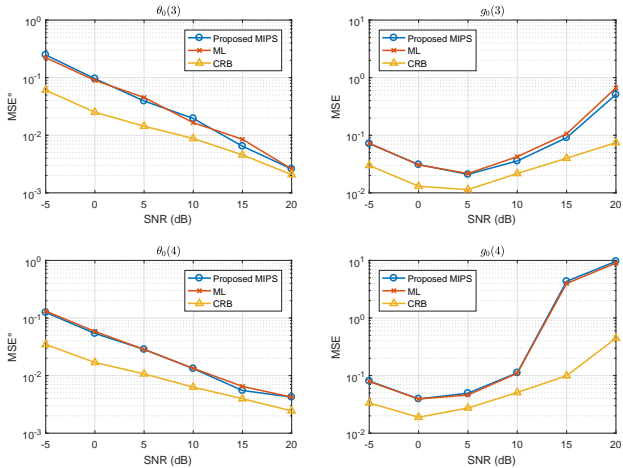}
	\caption{$\mathrm{MSE}(\theta_{0}(k))$ and $\mathrm{MSE}(g_{0}(k))$ vs. SNR for $M=16$, $N=12$, and $L=0$. Two realizations of $\theta_{0}$ and $g_{0}$ are given with the corresponding CRBs.}
\end{figure}

In Fig. 2 and Fig. 3, assuming that $\theta_{0}$ and $g_{0}$ are deterministic, the Cram\'er-Rao bound (CRB) of $\theta_{0}$ and $g_{0}$ is provided at different SNRs when $L=0$, which was derived in \cite{wang2017angular}. The $k$-th realizations of $\theta_{0}$ and $g_{0}$ are denoted by $\theta_{0}(k)$ and $g_{0}(k)$, respectively. The MSEs of the MIPS-based and ML channel estimators conditioned on $\theta_{0}(k)$ and $g_{0}(k)$, which are
\begin{align}
\mathrm{MSE}(\theta_{0}(k))&=\mathbb{E}\{(\hat{\theta}_{0}-\theta_{0})^{2}|\theta_{0}(k), g_{0}(k)\},\\
     \mathrm{MSE}(g_{0}(k))&=\mathbb{E}\{|\hat{g}_{0}-g_{0}|^{2}|\theta_{0}(k), g_{0}(k)\},
\end{align}
are plotted as well. From Fig. 2, observe that the gap between the MSEs of the MIPS-based and ML channel estimators are negligible as in Fig. 1. In addition, we point out that the gap between $\mathrm{MSE}(\theta_{0}(k))$ and the CRB is not increased, whereas the gap between $\mathrm{MSE}(g_{0}(k))$ and the CRB is increased as the SNR exceeds the medium SNR regime. This follows because one-bit ADCs convey only directional information, which is embedded in $\theta_{0}$; no magnitude information is provided, which is contained in $|g_{0}|$.

\section{Conclusion}
We proposed the MIPS-based channel estimator for large-scale antenna systems using one-bit ADCs where the array of antennas at the receiver is characterized by its angular geometry. The MIPS-based channel estimator finds the estimate of the dominant channel among multiple paths using two-stage channel estimation, which has low computational complexity. In the first stage, the DOA is estimated using the MIPS. After finding the DOA estimate, the second stage is performed, which finds the channel fading coefficient estimate using convex optimization. The simulation results showed that the MIPS-based channel estimator performs close to the pML channel estimator.

\section*{Acknowledgment}
This work was supported by the Institute for Information \& communications Technology Promotion (IITP) under grant funded by the MSIT of the Korea government (No.2018(2016-0-00123), Development of Integer-Forcing MIMO Transceivers for 5G \& Beyond Mobile Communication Systems).

\bibliographystyle{IEEEtran} 
\bibliography{refs_all}

\end{document}